\documentclass[twoside,11pt,final]{entics} 
\usepackage{enticsmacro}
\usepackage{graphicx}
\usepackage[all]{xy}
\usepackage{amsmath}
\makeatletter
\newcommand{\allignLabel}[1]{\refstepcounter{equation}(\theequation)\def\tmplab{#1}\ltx@label\tmplab}
\makeatother
\usepackage{anyfontsize}
\usepackage{proof}
\usepackage{mathrsfs}
\usepackage{stmaryrd}
\usepackage{dsfont}
\usepackage{cleveref}
\usepackage[all]{hypcap}
\usepackage{mathtools}
\usepackage{quiver}

\makeatletter
\newcommand\Label[1]{\refstepcounter{equation}(\theequation)\ltx@label{#1}}
\makeatother

\usepackage{color}
\def\bR{\begin{color}{red}}
\def\bB{\begin{color}{blue}}
\def\bM{\begin{color}{magenta}}
\def\bC{\begin{color}{cyan}}
\def\bW{\begin{color}{white}}
\def\bBl{\begin{color}{black}}
\def\bG{\begin{color}{green}}
\def\bY{\begin{color}{yellow}}
\def\e{\end{color}\xspace}
\newcommand{\bit}{\begin{itemize}}
\newcommand{\eit}{\end{itemize}\par\noindent}
\newcommand{\ben}{\begin{enumerate}}
\newcommand{\een}{\end{enumerate}\par\noindent}
\newcommand{\beq}{\begin{equation}}
\newcommand{\eeq}{\end{equation}\par\noindent}
\newcommand{\beqa}{\begin{eqnarray*}}
\newcommand{\eeqa}{\end{eqnarray*}\par\noindent}
\newcommand{\beqn}{\begin{eqnarray}}
\newcommand{\eeqn}{\end{eqnarray}\par\noindent}

\newcommand{\ie}{\emph{i.e.}}
\newcommand{\eg}{\emph{e.g.}}

\newcommand{\alt}{~\mid~}

\newcommand{\inl}[1]{\ensuremath{\mathtt{inj}_1{\;#1}}}
\newcommand{\inr}[1]{\ensuremath{\mathtt{inj}_2{\;#1}}}
\newcommand{\ini}[1]{\mathtt{inj}_i{\;#1}}
\newcommand{\pv}[2]{\ensuremath{\langle #1,#2 \rangle}}

\newcommand{\one}{\mathtt{I}}

\newcommand{\qubit}{\mathtt{qubit}}
\newcommand{\two}{\mathtt{2}}

\newcommand{\fold}[1]{\ensuremath{\mathtt{fold}{\;#1}}}

\newcommand{\letv}[3]{{\mathtt{let}}\,{#1}={#2}~{\mathtt{in}}~{#3}}

\newcommand{\ffix}{\ensuremath{\mathtt{fix}~}}
\newcommand{\ffixl}{\ensuremath{\mathtt{fix}^\lambda~}}

\newcommand{\entailiso}{\vdash_{\isoterm}}

\newcommand{\iso}{\ensuremath{\leftrightarrow}}

\newcommand{\isoterm}{\ensuremath{\omega}}

\newcommand{\nnext}[1]{\mathtt{next}{\;\!#1}}
\newcommand{\later}{{\ensuremath{\blacktriangleright}}}

\newcommand{\id}{\mathrm{id}}

\newcommand{\N}{{\mathbb N}}

\newcommand{\CC}{{\mathbf C}}

\newcommand{\NNe}{(\ensuremath{{{\mathbf C_{\mathrm{e}}})^\infty}}}
\newcommand{\NN}{{\ensuremath{{\mathbf C}^\infty_{\mathrm{e}}}}}
\newcommand{\QQ}{{\ensuremath{{\mathbf C}^\infty}}}
\newcommand{\QQdag}{{\ensuremath{{\mathbf C}^\infty_{\dagger}}}}

\newcommand{\Scat}{{\mathbf S}}

\newcommand{\defeq}{\stackrel{\textrm{{\scriptsize def}}}{=}}

\newcommand{\Set}{\ensuremath{\mathbf{Set}}}

\newcommand{\Hilb}{\mathbf{Hilb}}

\newcommand{\PInj}{\mathbf{PInj}}

\newcommand{\inim}{\textnormal{\textexclamdown}}

\newcommand\embed\hookrightarrow

\newcommand\natto\Rightarrow

\newcommand{\abs}[1]{\left\vert #1 \right\vert}

\newcommand{\dg}{^{\dagger}}
\newcommand{\inv}{^{-1}}
\newcommand{\op}{^{\mathrm{op}}}

\newcommand{\fix}{\mathrm{fix}}
\newcommand{\eval}{\mathrm{eval}}

\newcommand{\iid}{\mathrm{id}}
\newcommand{\comp}{\mathrm{comp}}

\newcommand{\nnoma}{\!^\noma}

\newcommand{\set}[1]{\ensuremath{\{#1\}}}

\newcommand{\den}[1]{\left\llbracket #1 \right\rrbracket}
\newcommand\sem\den

\newcommand{\ov}[1]{\overrightarrow{#1}}

\newcommand{\ket}[1]{\ensuremath{\left|  #1 \right\rangle}}

\newcommand{\yo}{\text{\usefont{U}{min}{m}{n}\symbol{'210}}}

\newcommand{\noma}{\text{\usefont{U}{min}{m}{n}\symbol{'005}}\!}

\DeclareFontFamily{U}{min}{}
\DeclareFontShape{U}{min}{m}{n}{<-> udmj30}{}

\newcommand{\secref}[1]{\S \ref{#1}}

\def\conf{MFPS 2026} 	%%Fill in the acronym for your conference (with year)
\volume{6}			%Fill in the ENTICS volume number here
\def\volu{6}			% and here
\def\papno{9}			%Fill in your paper number here
\def\mydoi{18683}
\def\lastname{Louis Lemonnier} %Lastnames appear in the running header 
\def\runtitle{Non-cartesian guarded recursion with daggers} %Short title appears in the running header on even pages. 
\def\copynames{L. Lemonnier}
\begin{document}
%%%Note the beginning and end of the frontmatter section that starts here%%%%%
\begin{frontmatter}
	\title{Non-cartesian Guarded Recursion With Daggers}
	%%%%%%%%%%%%%%%%%%%%%%%%%%%%			This Thanks is optional.
	%%%%Now the author(s) names(s)%%%%%
	\author{Louis Lemonnier}	%%Note NO SPACE between 
	%%%Next come the addresses%%%%
	\address{University of Edinburgh, UK}
	%%%Note: if both authors share same institution, only list the address once, after the second 
	%%%author. 
	%%%There also is a link from the first author to the co-author's address to show how to list 
	%%%affiliations to more than one institution, when needed. 

	\begin{abstract}
		Guarded recursion is a framework allowing for a formalisation of streams in
		classical (as opposed to concurrent, probabilistic, quantum) programming
		languages. The latter take their semantics in cartesian closed categories.
		However, some programming paradigms do not take their semantics in a
		cartesian setting; this is the case for concurrency, reversible and quantum
		programming for example. In this paper, we focus on reversible programming
		through its categorical model in dagger categories, which are categories
		that contain an involutive operator on morphisms. We show how to introduce
		the framework of guarded recursion into dagger categories with sufficient
		structure for data types, also called dagger rig categories. First, given
		an arbitrary category, we build another category shown to be suitable for
		guarded recursion in multiple ways, via enrichment and fixed point
		theorems. We then study the interaction between this construction and the
		structure of dagger rig categories, aiming for reversible programming.
		Finally, we show that our construction is suitable as a model of
		higher-order reversible programming languages, such as symmetric
		pattern matching, to which we add guarded recursion features.
	\end{abstract}
\end{frontmatter}

%%%%%%%%%%%%%%%%%%%%%%%%%%%%%%%%%%%%%%%%%%%%%%%%%%%%%%%%%%%%%%%%%%%%%%%%%%%%%%
\section{Introduction}
%%%%%%%%%%%%%%%%%%%%%%%%%%%%%%%%%%%%%%%%%%%%%%%%%%%%%%%%%%%%%%%%%%%%%%%%%%%%%

Most programming language take a sound and adequate interpretation in
\emph{cartesian closed} categories. However, cartesian closed categories do not
model all computing paradigms: reversible computation takes a sound and
adequate model \cite{kaarsgaard2021join,nous2021invcat,nous2024invrec} in
\emph{inverse categories} \cite{kastl1979inverse,kaarsgaard2017join}. The
category of sets and partial injections is a concrete such category, and is not
cartesian closed. Another field of interest where the categories are not
cartesian closed is \emph{quantum} computing: \eg~the category of
\emph{completely positive maps} \cite{selinger2007cpm}, which is \emph{compact}
closed, but not cartesian, followed by many others \cite{selinger2008fully,hasuo2011goi,pagani2014quantitative,clairambault2019game,jia2022variational,clairambault2019full,tsukada2024enriched}. The categories mentioned before are
models of mixed quantum computing. The ones that integrate pure quantum
computing are usually based on Hilbert spaces and bounded linear maps
\cite{sabry2018symmetric,heunen2022information}.

Among the categories mentioned above, inverse categories and categories based
on Hilbert spaces are \emph{dagger} categories. The dagger is an involutive
operation on morphisms. This means that dagger categories possess a sort of
partial inverse structure; for example, the partial inverse of a partial
function $f \colon \{ 0,1 \} \to \{ 0,1 \}$ (see on the left below) is also a
partial function $f\dg \colon \{ 0,1 \} \to \{ 0,1 \}$ (see on the right
below).
\begin{equation}
	\label{eq:pinj-fun}
	f(x) = \left\{
		\begin{array}{ll}
			1 & \text{if } x=0 \\
			\text{undefined} & \text{otherwise}
		\end{array}
		\right. \qquad
	f\dg (x) = \left\{
		\begin{array}{ll}
			0 & \text{if } x=1 \\
			\text{undefined} & \text{otherwise}
		\end{array}
		\right.
\end{equation}
Moreover, this partial inverse allows us to provide an interpretation to
functions even if the category is not closed \cite{kaarsgaard2021join,nous2021invcat,nous2024invrec}. A dagger category equipped with sufficient
structure to interpret basic data types ($\otimes$ for pairs, $\oplus$ for
control) is called a \emph{dagger rig} category. Dagger rig categories are
believed to be the canonical model of reversible
computation~\cite{carette2024compositional}.

While the notions of inductive types and recursion are well-studied in
cartesian closed categories~\cite{abramsky95domain,fiore-phd}, this is not the
case for dagger categories. Some families of dagger categories, such as inverse
categories, are suitable to model inductive types and recursion, through their
enrichment~\cite{kaarsgaard2017join} in directed complete partial orders -- a
cartesian closed structure capable of solving recursive domain
equations~\cite{domain-theory}. Inverse categories being in general not closed,
the notion of enrichment is central: terms and morphisms typing judgements are
interpreted in the enrichment and not in the externalisation, \emph{i.e.}~the
underlying inverse category. The enrichment in dcpos provides a fixed point
operator, therefore providing a model of reversible programming languages which
are Turing complete~\cite{nous2024invrec}. However, the story does not go so
well for Hilbert spaces, which lack the proper enrichment~\cite[Proposition
2.10]{heunen2013l2}. Even if inductive types can be interpreted in Hilbert
spaces~\cite[Theorem 3.2]{barr92compact}, it is an open question whether they
can interpret recursion. We choose not to tackle this open question in this
paper, and we instead suggest a different path, through guarded recursion.

Guarded recursion allows us to capture and control recursive calls within the
type system~\cite{nakano2000recursion}, with the help of a modality, written
$\later$ and called \emph{later}. While this framework has been extensively
studied for classical computation (as opposed to \eg~reversible or quantum
computation) in cartesian closed settings, it is yet to be applied to less
traditional ways of computation.

In the search for a denotational semantics of guarded recursion, a first
account of solution to solve guarded domain equations was given within
sheaves~\cite{gianantonio2003unifying,gianantonio2004unifying}. Then an
adequate denotational semantics was given with ultrametric
spaces~\cite{birkedal2010metric}, and later the same authors provided a more
general semantics within the \emph{topos of trees}~\cite{birkedal2012first}.
The topos of trees is a category that is cartesian closed, which means in
particular that it is a model of the simply-typed $\lambda$-calculus and of
most traditional programming languages.

Several papers have then used synthetic guarded domain theory to develop
refined models of guarded types or guarded recursion \cite{birkedal2016guarded,mannaa2020ticking}. Guarded recursion has also been generalised to the case of
causal structures~\cite{basold2023causal} in a cartesian closed setting.

In the literature, guarded traced categories~\cite{goncharov2018guarded} are an
instance of non-cartesian study of guarded recursion. The construction in this
paper does not necessary yield a guarded trace category, offering a different
angle in generalising guarded recursion.

The goal of this paper is to show that guarded recursion has a use beyond the
formalisation of streams in classical programming languages, in particular for
reversibility, through dagger categories. The study of the dagger structure is
central in forming a link between the semantics of reversible programming
languages and the semantics of guarded recursion in the topos of trees.
Starting from an arbitrary category, we show that we can construct a guarded
model out of this category, enriched in the topos of trees, therefore allowing
for a model of guarded fixed points. Aiming at reversible programming, we study
the interaction between the \emph{dagger rig} structure and our guarded
construction, in order to use the latter as interpretation for a reversible
programming language with guarded recursion.

% %%%%%%%%%%%%%%%%%%%%%%%%%%%%%
% \subsection*{Content of the paper}
% \label{sub:guar-related}
% %%%%%%%%%%%%%%%%%%%%%%%%%%%%%

\emph{Content of the paper.}
The paper starts with a brief summary of classical guarded recursion. We
picture the corresponding syntax (see~\secref{sub:classic-syntax}) through the
guarded $\lambda$-calculus~\cite{birkedal2016guarded}, and then we introduce
its interpretation in the topos of trees (see \secref{sub:topos}). We then
present the \emph{guarded construction} (see \secref{sec:construction}):
starting from an arbitrary category with a terminal object, we construct a
categorical model of guarded recursion. We show that the constructed category
has characteristics similar to the topos of trees (see
\secref{sub:dagger-guarded}), and that it is even a model of guarded fixed
points (see \secref{sub:guarded-fixed}), allowing for the interpretation of
guarded inductive types. In the following section, we focus on categorical
models of reversible programming. We recall the definition of a dagger rig
category (see \secref{sub:dagger-rig}). We show that the guarded construction
preserves the rig structure (see \secref{sub:rig-construction}), but we need to
refine the construction to study its potential dagger structure (see
\secref{sub:dagger-guarded-dagger}). Finally, we show that the guarded
construction from a dagger rig category is a categorical model of a reversible
programming language with guarded recursion (see \secref{sec:application}).

% %%%%%%%%%%%%%%%%%%%%%%%%%%%%%
% \subsection*{Proofs}
% %%%%%%%%%%%%%%%%%%%%%%%%%%%%%

% We provide proofs for the statements that embody our contribution in appendix
% (see \secref{app:proofs}).

%%%%%%%%%%%%%%%%%%%%%%%%%%%%%%%%%%%%%%%%%%%%%%%
\section{Classical Guarded Recursion}
\label{sec:classic}
%%%%%%%%%%%%%%%%%%%%%%%%%%%%%%%%%%%%%%%%%%%%%%%

We provide an introduction to the concepts of guarded recursion, first through
the syntax~\cite{birkedal2016guarded} (see \secref{sub:classic-syntax}), and
then through its denotational semantics (see \secref{sub:topos}) in the
\emph{topos of trees}~\cite{birkedal2012first}.

%%%%%%%%%%%%%%%%%%%%%%
\subsection{Syntactically}
\label{sub:classic-syntax}
%%%%%%%%%%%%%%%%%%%%%%

The initial goal of guarded recursion \cite{nakano2000recursion} is to ensure
that functions defined on coinductive data types (\eg~streams) are
well-defined. This is done with a modality, called \emph{later}, and recursive
calls can only be nested under the constructor associated to this modality.

We use the symbol $\later$ in the syntax for this modality. Semantically, this
modality is represented by a functor (see \secref{sub:topos}). An excerpt of a
type system is given below, where the statement `$X$ guarded in $A$' means that
$X$ is under a later modality $\later$ in the expression, \eg~in the type $1
\oplus \later X$, which gives the type of guarded natural numbers $\mu X . 1
\oplus \later X$. Given a type $A$, we can form the type of lists as $\mu X .
1 \oplus (A \otimes \later X)$, which we sometimes write $[A]$.
\[
	\infer{\Theta, X \vdash X}{}
	\qquad
	\infer{\Theta \vdash \later A}{\Theta \vdash A}
	\qquad
	\infer{\Theta \vdash \mu X . A}{\Theta, X \vdash A & X \text{ guarded in } A}
\]

This modality must be introduced into the terms in the syntax. In a
$\lambda$-calculus fashion, such as in the guarded $\lambda$-calculus
\cite{birkedal2016guarded}, we write $\nnext\!$ for the constructor introducing
the later modality. Here is an excerpt of the typing rules.
\[
	\infer{
		\Gamma \vdash \nnext M \colon \later A
	}{
		\Gamma \vdash M \colon A
	}
	\qquad
	\infer{
		\Gamma \vdash \ini M \colon A_1 \oplus A_2
	}{
		\Gamma \vdash M \colon A_i
	}
	\qquad
	\infer{
		\Gamma \vdash \fold M \colon \mu X . A
	}{
		\Gamma \vdash M \colon A[\mu X . A / X]
	}
\]

\begin{example}
	\label{ex:lists}
	Let us fix a type $A$. Within the type of lists $[A] = \mu X . 1 \oplus (A
	\otimes \later X)$, the empty list is obtained with $\fold{} (\inl *)$,
	which we will shorten as $[~]$.  The syntax for the list with one element
	$M$ of type $A$ is $\fold{} (\inr{} (M \otimes (\nnext{} \inl *)))$,
	shorten as $M :: \nnext [~]$.
\end{example}

One of the key points of guarded recursion is fixed points. The usual fixed
point operator, such as for a typed $\lambda$-calculus~\cite{plotkin1977lcf,gunter1992semantics,fiore-phd}, has the following type:
\(
	\ffixl \colon (A \to A) \to A.
\)
The corresponding operational rule usually \emph{unfolds} the fixed
point:
\(
	\ffixl M \to M(\ffixl M).
\)

In guarded recursion, the fixed point is more restrictive and recursive calls
are \emph{guarded}. To do so, the type of a guarded fixed point is $\ffix
\colon (\later A \to A) \to A$, and given $\cdot \vdash M \colon \later A
\to A$, that is to say, a function that takes a guarded term as an input,
we have the following operational rule:
\(
	\ffix M \to M(\nnext \ffix M).
\)
This means that after applying the fixed point rule, the next application must
happen one step `in the future'.

Studying the operational behaviour of the terms is an important aspect of
programming language theory, but it is not the focus of this paper. In
particular, there are many interesting notions in guarded type theory that do
not apply in this paper, such as productivity. We restrict ourselves to a
\emph{structural} perspective on programming languages, by analysing their
categorical denotational semantics.

The syntax for guarded recursion, introduced above, admits a sound and adequate
denotational semantics~\cite{birkedal2016guarded} in the \emph{topos of trees},
which is a cartesian closed category with sufficient structure for guarded
recursion, as explained in the next section.

%%%%%%%%%%%%%%%%%%%%%%
\subsection{Categorically}
\label{sub:topos}
%%%%%%%%%%%%%%%%%%%%%%

In this section, we recall the main categorical structure underlying guarded
recursion, referred to as \emph{synthetic guarded domain theory}
\cite{birkedal2012first}. We write $\Scat$ for the category $\Set^{\N^{\mathrm
op}}$, referred to as the \emph{topos of trees}. The definition of this
category involves the natural numbers $\N$ with the usual order, which is
pictured below. We also picture its opposite category, $\N^{\mathrm op}$.
\[
	\N ::
	\begin{tikzcd}
		0 & 1 & 2 & 3 & \cdots
		\arrow["\leq", from=1-1, to=1-2]
		\arrow["\leq", from=1-2, to=1-3]
		\arrow["\leq", from=1-3, to=1-4]
		\arrow["\leq", from=1-4, to=1-5]
	\end{tikzcd}
	\quad
	\N^{\mathrm op} ::
	\begin{tikzcd}
		0 & 1 & 2 & 3 & \cdots
		\arrow["\leq"', from=1-2, to=1-1]
		\arrow["\leq"', from=1-3, to=1-2]
		\arrow["\leq"', from=1-4, to=1-3]
		\arrow["\leq"', from=1-5, to=1-4]
	\end{tikzcd}
\]

The notation $\Set^{\N^{\mathrm op}}$ means that the objects of the category
are functors $\N^{\mathrm op} \to \Set$. A functor $X \colon \N^{\mathrm op}
\to \Set$ assigns to every object $n$ of $\N$ (\ie~a natural number) a set
$X(n)$ and to every morphism of $\N$ (such as $n \leq n{+}1$) a function $r^X_n$
between the sets $X(n{+}1)$ and $X(n)$, and therefore $X$ can be pictured in the
category $\Set$:
\[
	\begin{tikzcd}
		X(0) & X(1) & X(2) & X(3) & \cdots
		\arrow["r^X_0"', from=1-2, to=1-1]
		\arrow["r^X_1"', from=1-3, to=1-2]
		\arrow["r^X_2"', from=1-4, to=1-3]
		\arrow[from=1-5, to=1-4]
	\end{tikzcd}
\]

Finally, morphisms in the category $\Scat$ are natural transformations between
the functors $\N^{\mathrm op} \to \Set$. If $X$ and $Y$ are functors
$\N^{\mathrm op} \to \Set$, a natural transformation $f \colon X \to Y$ is
pictured in $\Set$ with the following commutative diagram:
\[
	\begin{tikzcd}
		X(0) & X(1) & X(2) & X(3) & \cdots \\
		Y(0) & Y(1) & Y(2) & Y(3) & \cdots
		\arrow["r^X_0"', from=1-2, to=1-1]
		\arrow["r^X_1"', from=1-3, to=1-2]
		\arrow["r^X_2"', from=1-4, to=1-3]
		\arrow[from=1-5, to=1-4]
		\arrow["r^Y_0", from=2-2, to=2-1]
		\arrow["r^Y_1", from=2-3, to=2-2]
		\arrow["r^Y_2", from=2-4, to=2-3]
		\arrow[from=2-5, to=2-4]
		\arrow["f_0", from=1-1, to=2-1]
		\arrow["f_1", from=1-2, to=2-2]
		\arrow["f_2", from=1-3, to=2-3]
		\arrow["f_3", from=1-4, to=2-4]
	\end{tikzcd}
\]

This category is a \emph{topos}, which means that it contains a lot of
structure. In particular, it is a cartesian category, and the cartesian product
is obtained pointwise and directly inherited from the one in $\Set$. Moreover,
like any topos~\cite[\S I.6, Proposition 1]{maclane2012sheaves}, the category
$\Scat$ is closed.

In particular, the exponential $[X \to Y](-) \colon \N^{\mathrm op} \to \Set$
is given by $\Scat(\yo(-) \times X, Y)$ where $\yo \colon \N^{\mathrm op} \to
\Scat$ is the Yoneda embedding (the use of the Japanese hiragana \emph{``yo''}
was democratised by Loregian \cite{loregian2021coend}). In our case, the Yoneda
embedding can be described as follows: given a natural number $n$, the functor
$\yo(n) \times X$ is the same as $X$ but truncated after the $n^{\text{th}}$
component. It is pictured in $\Set$ as:
\(
	\begin{tikzcd}
		X(0) & \cdots & X(n) & \emptyset & \cdots
		\arrow["r^X_0"', from=1-2, to=1-1]
		\arrow["r^X_{n-1}"', from=1-3, to=1-2]
		\arrow["\inim"', from=1-4, to=1-3]
		\arrow[from=1-5, to=1-4]
	\end{tikzcd}
\)
with $\inim_X \colon \emptyset \to X$ being the unique map from the empty set.
The elements of $\Scat(\yo(n) \times X, Y)$ are thus truncated natural
transformations,
such as:
\begin{equation}
	\label{eq:truncated-nt}
	\begin{tikzcd}[column sep=2em]
		X(0) & \cdots & X(n) & \emptyset & \cdots \\
		Y(0) & \cdots & Y(n) & Y(n{+}1) & \cdots
		\arrow["r^X_0"', from=1-2, to=1-1]
		\arrow["r^X_{n-1}"', from=1-3, to=1-2]
		\arrow["\inim"', from=1-4, to=1-3]
		\arrow[from=1-5, to=1-4]
		\arrow["r^Y_0", from=2-2, to=2-1]
		\arrow["r^Y_{n-1}", from=2-3, to=2-2]
		\arrow["r^Y_{n}", from=2-4, to=2-3]
		\arrow[from=2-5, to=2-4]
		\arrow["f_0", from=1-1, to=2-1]
		\arrow["f_n", from=1-3, to=2-3]
		\arrow["\inim", from=1-4, to=2-4]
	\end{tikzcd}
\end{equation}
and therefore can be described as a finite family of functions between sets
$(f_0, \dots, f_n)$ and the function between sets $r^{[X \to Y]}_n \colon
\Scat(\yo({{n{+}1}}) \times X, Y) \to \Scat(\yo(n) \times X, Y)$ gets a family
of ${n{+}1}$ functions and drops the last one.

Because $\Scat$ is a cartesian closed category, it is a model of the
simply-typed $\lambda$-calculus. In the following, we introduce the structure
related to guarded recursion and show that it allows us to model a
\emph{guarded} fixed point operator.

\begin{definition}[Later functor in $\Scat$ {\cite[\S 2.1]{birkedal2012first}}]
	\label{def:later-s}
	The later functor $L\colon \Scat \to \Scat$ is such that given an object
	$X$ in $\Scat$, we have $LX(0) = 1$ (the terminal object) and $LX(n{+}1) =
	X(n)$; and given a morphism $\alpha\colon X \to Y$ in $\Scat$, we have
	$L\alpha_0 = ~!_1$ (the terminal map), and $(L\alpha)_{n{+}1} = \alpha_n$.
\end{definition}

Simply put, this functor shifts all components to the right. If $X$ is an
objects in $\Scat$, the object $LX$ is pictured in $\Set$ as follows:
\(
	\begin{tikzcd}
		1 & X(0) & X(1) & X(2) & \cdots
		\arrow["!"', from=1-2, to=1-1]
		\arrow["r^X_0"', from=1-3, to=1-2]
		\arrow["r^X_1"', from=1-4, to=1-3]
		\arrow[from=1-5, to=1-4]
	\end{tikzcd}
\)
When we use a category as a model for a programming language, the objects are
models for types, whereas morphisms are models for terms and functions in the
language. The \emph{later} functor can be interpreted in the syntax as a
modality. We then show that this modality can be introduced by terms in the
syntax.

\begin{definition}[Next in $\Scat$ {\cite[\S 2.1]{birkedal2012first}}]
	\label{def:next-s}
	We call \emph{next} the natural transformation $\nu\colon id\natto L$ in
	$\Scat$ defined as $\nu_{X,0} = ~!_1$ and $\nu_{X,n{+}1} = r^X_n$ for all
	objects $X$.
\end{definition}

This category $\Scat$ comes with a specific form of fixed point for terms. In
the case of the $\lambda$-calculus for example, a fixed point operator is a
family of morphisms $\fix^\lambda_X \colon [X \to X] \to X$, such that
$\eval_{X,X} \circ \pv{\iid_{[X \to X]}}{\fix^\lambda_X} = \fix^\lambda_X$,
where $\eval_{X,Y} \colon [X \to Y] \times X \to Y$ is the morphism that
evaluates its first argument in its second argument. In $\Scat$, we have a
different kind of fixed point operator, because it involves the later modality.
We refer to it as the \emph{guarded} fixed point operator.

\begin{lemma}[Guarded fixed points {\cite[\S 2.3]{birkedal2012first}}]
	\label{lem:g-fix-point}
	In the category $\Scat$, there exists a family of morphisms $\fix_X \colon
	[LX \to X] \to X$ such that $\eval_{LX,X} \circ \pv{\iid_{[LX \to
	X]}}{\nu_X \circ \fix_X} = \fix_X$.
\end{lemma}

Fixed points can actually be obtained~\cite{birkedal2012first} for a more
general family of morphisms, that we call \emph{contractive} morphisms. A
morphism is contractive if it factorises through a component of the \emph{next}
natural transformation. Contractivity is stable under composition (even with a
non-contractive morphism), by tensor product and by currying.

Besides fixed points, recursive domain equations can also be solved in $\Scat$
\cite{birkedal2012first} for a certain class of functors -- roughly the ones
that involve the later functor, which we call \emph{contractive} functors.  We
do not provide full detail on this point, but we show later that the categories
we introduce -- namely, categories for non-cartesian guarded recursion -- can
solve domain equations for inductive types in a similar manner as in $\Scat$.

%%%%%%%%%%%%%%%%%%%%%%%%%%%%%%%%%%%%%%%%%%%%%%%
\section{Non-Cartesian Guarded Construction}
\label{sec:construction}
%%%%%%%%%%%%%%%%%%%%%%%%%%%%%%%%%%%%%%%%%%%%%%%

This section aims at providing a categorical tool sufficient to model guarded
recursion in a non-cartesian setting (see Examples~\ref{ex:pinj} and
\ref{ex:hilb}). We exhibit a construction similar to the category $\Scat$ (see
\secref{sub:topos}) which allows for the interpretation of guarded recursion.
We also show that guarded domain equations can be solved in this setting.
We work with categories of the form $\CC^{\N^{\mathrm op}}$, where $\N$ is
the category of natural numbers starting from $0$, with morphisms defined by
the usual order on natural numbers, as in \secref{sub:topos}. Given an object $X$ of
$\CC^{\N^{\mathrm op}}$, that is, a functor $\N^{\mathrm op}\to \CC$, its image
on the morphism $n \leq n{+}1$ is written $r^X_n \colon X(n{+}1) \to X(n)$, and is
a morphism in $\CC$. This object $X$ of $\CC^{\N^{\mathrm op}}$ can be
represented with a diagram:
\begin{equation}
	\label{eq:diagram-view}
	\begin{tikzcd}
		X(0) & X(1) & X(2) & X(3) & \cdots
		\arrow["r^X_0"', from=1-2, to=1-1]
		\arrow["r^X_1"', from=1-3, to=1-2]
		\arrow["r^X_2"', from=1-4, to=1-3]
		\arrow[from=1-5, to=1-4]
	\end{tikzcd}
\end{equation}
and a morphism $f \colon X \to Y$ can be pictured with the following diagram
in $\CC$:
\begin{equation*}
	\label{eq:diagram-view-natural}
	\begin{tikzcd}
		X(0) & X(1) & X(2) & X(3) & \cdots \\
		Y(0) & Y(1) & Y(2) & Y(3) & \cdots
		\arrow["r^X_0"', from=1-2, to=1-1]
		\arrow["r^X_1"', from=1-3, to=1-2]
		\arrow["r^X_2"', from=1-4, to=1-3]
		\arrow[from=1-5, to=1-4]
		\arrow["r^Y_0", from=2-2, to=2-1]
		\arrow["r^Y_1", from=2-3, to=2-2]
		\arrow["r^Y_2", from=2-4, to=2-3]
		\arrow[from=2-5, to=2-4]
		\arrow["f_0", from=1-1, to=2-1]
		\arrow["f_1", from=1-2, to=2-2]
		\arrow["f_2", from=1-3, to=2-3]
		\arrow["f_3", from=1-4, to=2-4]
	\end{tikzcd}
\end{equation*}
The category $\CC^{\N\op}$ looks like a category of presheaves, but is not a
topos or even cartesian closed in general. We also show in the following that
it is not a dagger category. However, it does not prevent us from using the
dagger from the underlying category $\CC$.

Let us fix a category $\CC$ with terminal object $T$. For all objects $A$,
there is a unique morphism $!_A \colon A \to T$. We write $\QQ$ for the
category $\CC^{\N\op}$, that we refer to as the \emph{guarded construction}.
This construction preserves terminal objects (and in fact, all limits and
colimits).

\begin{lemma}
	\label{lem:guarded-terminal}
	The object of $\QQ$ pictured in $\CC$ as:
	\(
		\begin{tikzcd}
			T & T & T & T & \cdots
			\arrow["\iid_T"', from=1-2, to=1-1]
			\arrow["\iid_T"', from=1-3, to=1-2]
			\arrow["\iid_T"', from=1-4, to=1-3]
			\arrow[from=1-5, to=1-4]
		\end{tikzcd}
	\)
	is a terminal object in $\QQ$.
\end{lemma}

We now show how the guarded construction is related to the topos of trees
$\Scat$.

Categories in computer science are usually \emph{locally small}, meaning that
given two objects $A$ and $B$, there is a \emph{set} of morphisms $A\to B$.
Enrichment is the study of the structure of those sets of morphisms, which
could be vector spaces or topological spaces for example (more detail can be
found in~\cite{KELLY196515,kelly1982basic,maranda_1965}). It turns out that
homsets of $\QQ$ can be seen as objects in~$\Scat$ -- \ie~sequences of sets
with morphisms between them. The results on~$\Scat$ presented above (see
\secref{sub:topos}) can then be applied at the level of morphisms in $\QQ$.

\begin{lemma}
	\label{lem:enriched}
	The category $\QQ$ yields an $\Scat$-enriched category with the following
	data:
	\begin{itemize}
		\item objects are objects in $\CC^{\N\op}$;
		\item hom-objects $\QQ(X,Y)$ of $\Scat$, defined as: $\QQ(X,Y)(n) = \{
			(f_0, \dots, f_n) \mid f \colon X \to Y \text{ in } \CC^{\N\op}
			\}$, a set of truncated natural transformations, and $\QQ(X,Y)(n{+}1)
			\to \QQ(X,Y)(n)$ is given as the function that forgets the last
			element;
		\item for all objects $X$, a morphism $\iid_X \colon 1 \to \QQ(X,X),$
			which outputs truncated identity natural transformations;
		\item for all objects $X, Y, Z$, a morphism:
			\( \comp_{X,Y,Z} \colon \QQ(Y,Z) \times \QQ(X,Y) \to \QQ(X,Z), \)
			which composes the truncated natural transformations.
	\end{itemize}
\end{lemma}

As we do not need much enriched category theory, we choose to keep the
conversation at the level of category -- where $\QQ$ has objects $X,Y$ and
morphism $f \colon X \to Y$, but we keep in mind the notation above for
$\QQ(X,Y)$ as an object in $\Scat$ and the corresponding morphisms in $\Scat$
for identity and composition.

In the categorical semantics of a programming language, we provide an
interpretation of types as objects in the category and of terms as morphisms.
In our construction, the morphisms live in $\QQ(X,Y)$, which is an object of
$\Scat$, so the semantics of the terms is given in a mathematical setting
sufficient to interpret guarded recursion. The enrichment described above
therefore can equip any categorical model $\CC$ with guarded recursion, and
$\QQ$ is also a model, assuming that the guarded construction $\QQ$ preserves
the structure of $\CC$ relevant to the semantics. We give an example of the
kind of structure that the guarded construction preserves with dagger rig
categories later in the paper (see \secref{sec:constructing}).

%%%%%%%%%%%%%%%%%%%%%%%%%%%%%
\subsection{The Guarded Structure of the Guarded Construction}
\label{sub:dagger-guarded}
%%%%%%%%%%%%%%%%%%%%%%%%%%%%%

A feature of categories of the form $\CC^{\N^{\mathrm op}}$, when $\CC$ has a
terminal object, is the later functor. Operationally, this functor works as a
delay modality (see \secref{sub:classic-syntax}). It can be used to keep track
of the depth of a term and the number of recursive calls. It shifts the
diagrammatic view in $\CC$ of an object $X$ (see Diagram~\ref{eq:diagram-view})
a step to the right, and adds a terminal object on the left.

\begin{definition}[Later functor]
	\label{def:later}
	The later functor $L^\CC \colon \QQ \to \QQ$ is such that for all objects
	$X$ in $\QQ$, we have $L^\CC X(0) = T$ (the terminal object) and $L^\CC
	X(n{+}1) = X(n)$; and for all morphisms $f \colon X \to Y$, we have $(L^\CC
	f)_0 = ~!_{T}$ and $(L^\CC f)_{n{+}1} = f_n$.
\end{definition}

We use the same letter $L$ for any later functor when it is not ambiguous.  If
ambiguity arises, we use the notation $L^\Set \colon \Scat \to \Scat$ (because
$\Scat = \Set^\infty$) and $L^\CC$.

Let us recall that an $\Scat$-enriched functor is a functor whose action on
morphisms are morphisms in $\Scat$. We see in the next section that being
$\Scat$-enriched is one of the necessary conditions for a functor to admit a
fixed point.

\begin{lemma}
	\label{lem:later-enriched}
	The functor $L^\CC \colon \QQ \to \QQ$ is $\Scat$-enriched.
\end{lemma}

The action on objects of the later functors are linked due to the enrichment.

\begin{lemma}
	\label{lem:later-enriched-view}
	If $X$ and $Y$ are two objects in $\QQ$, then we have $L^\Set \QQ(X,Y)
	\cong \QQ(L^\CC X, L^\CC Y)$.
\end{lemma}

Like in the category $\Scat$, the delay embodied by the functor $L$ can be
introduced by a natural transformation, called \emph{next}. This natural
transformation helps us introduce the delay in a programming language, as the
denotational semantics of a delayed program.

\begin{definition}[Next]
	\label{def:next}
	The \emph{next} natural transformation $\nu^\CC \colon \iid \natto L^\CC$
	is such that if $X$ is an object in $\QQ$, we have $\nu^\CC_{X,0} =
	~!_{X(0)}$ and $\nu^\CC_{X,n{+}1} = r^X_n$. It can be observed as a
	commutative diagram in $\CC$, where it maps the sequence representing $X$
	to the sequence of $L^\CC X$ as follows:
	\[\begin{tikzcd}
		X(0) & X(1) & X(2) & X(3) & \cdots \\
		T & X(0) & X(1) & X(2) & \cdots
		\arrow["r^X_0"', from=1-2, to=1-1]
		\arrow["r^X_1"', from=1-3, to=1-2]
		\arrow["r^X_2"', from=1-4, to=1-3]
		\arrow[from=1-5, to=1-4]
		\arrow["!"', from=1-1, to=2-1]
		\arrow["r^X_0", from=1-2, to=2-2]
		\arrow["r^X_1", from=1-3, to=2-3]
		\arrow["r^X_2", from=1-4, to=2-4]
		\arrow["!", from=2-2, to=2-1]
		\arrow["r^X_0", from=2-3, to=2-2]
		\arrow["r^X_1", from=2-4, to=2-3]
		\arrow[from=2-5, to=2-4]
	\end{tikzcd}\]
\end{definition}

\begin{remark}
	Similarly to the later functor, the natural transformation $\nu$ can be
	defined in $\Scat$ (see Definition~\ref{def:next-s}) as well as in $\QQ$.
	If any ambiguity arises, we use the notation $\nu^\Set \colon \id_\Scat
	\natto L^\Set$ and $\nu^\CC \colon \id_\QQ \natto L^\CC$.
\end{remark}

\begin{lemma}
	\label{lem:next-constructed}
	If $X$ and $Y$ are two objects in $\QQ$, then we have $\nu^\Set_{\QQ(X,Y)}
	\cong L^\CC_{X,Y}$.
\end{lemma}

%%%%%%%%%%%%%%%%%%%%%%%%%%%%%
\subsection{Fixed points of Locally Contractive Functors}
\label{sub:guarded-fixed}
%%%%%%%%%%%%%%%%%%%%%%%%%%%%%

We start by making precise what we mean by fixed point of a functor.

\begin{definition}[Fixed Point]
	\label{def:fixed-point}
	Given an endofunctor $T \colon \QQ \to \QQ$, a fixed point of $T$ is a pair
	$(X,\alpha \colon TX \to X)$ such that $\alpha$ is an isomorphism.
\end{definition}

A \emph{locally contractive} $\Scat$-functor is one whose morphism mapping is a
contractive morphism in $\Scat$ (see \secref{sub:topos}). We show later that a
locally contractive functor has a unique fixed point, up to isomorphism.  The
functor $L^\CC \colon \QQ \to \QQ$ is an example of locally contractive
functor;
and given two $\Scat$-enriched functors $F,G \colon \QQ \to \QQ$ such that
either $F$ or $G$ is locally contractive, then $FG$ is locally
contractive~\cite[Lemma 7.3]{birkedal2012first}.

Given a morphism $\alpha\colon X \to Y$ in $\QQ$, one says that
$\alpha$ is an $n$-isomorphism if the first $n$ components of $\alpha$
(that is to say, $\alpha_0,\dots,\alpha_{n-1}$) are isomorphisms.

\begin{lemma}
	\label{lem:niso}
	A locally contractive functor maps an $n$-isomorphism to an
	$n{+}1$-isomorphism.
\end{lemma}

We can now prove the next theorem, similar to the one in the
literature~\cite{birkedal2012first}.

\begin{theorem}
	\label{th:fixedpoint}
	If an endofunctor is locally contractive, then it has a fixed point.
\end{theorem}

We generalise the notion of locally contractive functors: an $\Scat$-functor $G
\colon (\QQ)^k \to \QQ$ is \emph{locally contractive} if it is locally
contractive on all variables.

\begin{theorem}[Parameterised Fixed Point]
	\label{th:param}
	A locally contractive functor admits a parameterised fixed point. In
	detail, if we have a locally contractive functor $G \colon (\QQ)^{k+1} \to
	\QQ$, there is a pair $(G^\noma,\phi^G)$ such that:
	\begin{itemize}
		\item $G^\noma \colon (\QQ)^k \to \QQ$ is a locally contractive functor,
		\item $\phi^G \colon G \circ \pv{\mathrm{id}}{G^\noma} \natto G^\noma$ is a
			natural isomorphism,
		\item for every object $\ov F$ in $(\QQ)^k$,
			the pair $(G^\noma \ov F,\phi^G)$ is the fixed point of $G(\ov F,-)$.
	\end{itemize}
\end{theorem}
% \begin{proof}[Proof of Theorem~\ref{th:param}]
% 	The proof is roughly the same as the one for categories enriched in models
% 	of guarded recursive terms \cite[Theorem 7.5]{birkedal2012first}; we recall
% 	the key points here, with our own notation. Given $\ov F$ an object in
% 	$(\QQ)^k$, the functor $T(\ov F,-) \colon \QQ \to \QQ$ is contractive, and
% 	thus has a fixed point $(\Omega(\ov F),\alpha^{\ov F})$ (see
% 	Theorem~\ref{th:fixedpoint}). The next step is to prove that the statement
% 	$\Omega(-)$ induces a functor $(\QQ)^k \to \QQ$.  Remember that we have
% 	$\Omega(\ov F)(n) = T(\ov F, -)^{n{+}1} Z(n)$; and given $\beta \colon \ov F
% 	\natto \ov G$, we have $\Omega(\beta)_n = (T(\beta,-)^{n{+}1}Z)_n$, which
% 	makes $\Omega(-)$ a functor: it preserves the identity and composition
% 	because $T$ does. Also, in the formula, $T$ is applied at least once and is
% 	locally contractive, thus the functor $\Omega \colon (\QQ)^k \to \QQ$ is
% 	locally contractive. The natural transformation
% 	$T\circ\pv{\mathrm{id}}{\Omega}\natto \Omega$ is obtained by looking at the
% 	square diagram in the first part of the proof; its isomorphic nature is
% 	inherited from Lemma~\ref{lem:niso}.
% \end{proof}

% Note that a parameterised fixed point provides a natural isomorphism, whose
% components are also isomorphisms.

\begin{remark}
	A stronger theorem can be stated~\cite[Theorem 7.5]{birkedal2012first},
	which allows for the interpretation of recursive dependent types. In this
	paper, we work within categories that are not necessarily monoidal closed;
	therefore we choose our type system to have a single polarity, and we
	restrict ourselves to (co)inductive types instead of general recursive
	types.
\end{remark}

We have shown that inductive domain equations can be solved in $\QQ$, without
any specific assumption on $\CC$. In the next section, we focus on reversible
programming through its model in dagger rig categories, where \emph{rig} stands
for the data structure of tensors and sum types, and the \emph{dagger} ensures
that operations are reversible. We later use Theorem~\ref{th:param} with a
subcategory of $\CC$ (see \secref{sub:sem-ground-types}), in order to have a
better interaction with the dagger structure.

%%%%%%%%%%%%%%%%%%%%%%%%%%%%%%%%%%%%%%%%%%%%%%%
\section{Dagger and Rig Structure in Guarded Models}
\label{sec:constructing}
%%%%%%%%%%%%%%%%%%%%%%%%%%%%%%%%%%%%%%%%%%%%%%%

In this section, we recall the main definitions and present some examples of
dagger rig categories (see \secref{sub:dagger-rig}). We then study the effect
of the guarded construction on a dagger rig category with a terminal object
(that turns out to be a zero object because of the dagger). The rig structure
generalises well to the guarded constructed category (see
\secref{sec:construction}), but not the dagger (see Remark~\ref{rem:dagger}).
In order to seize the interaction with the dagger structure, we refine our
construction and show that our construction contains a sufficiently big
subcategory that is also dagger rig (see \secref{sub:dagger-guarded-dagger}).
This category is then used in the rest of the paper, in particular to interpret
reversible functions (see \secref{sub:fo-fun}).

We study the interaction between the construction above and dagger rig
categories.

%%%%%%%%%%%%%%%%%%%%%%%%%%%%%
\subsection{Background: Dagger Rig Categories}
\label{sub:dagger-rig}
%%%%%%%%%%%%%%%%%%%%%%%%%%%%%

We start with the notion of \emph{dagger}. A category $\CC$ is a \emph{dagger
category} if there is a contravariant endofunctor $(-)\dg \colon \CC\op \to
\CC$ such that $X\dg = X$ for all objects $X$ and $f\dg\!\dg = f$ for all
morphisms $f$ in $\CC$. A functor $F$ between dagger categories is a
\emph{dagger functor} if $F(f\dg) = F(f)\dg$ for all morphisms $f$.  As a model
of computation, the dagger means that there is a sound way to \emph{reverse}
programs. Note that `reversibility' in the context of reversible programming
does not mean that all programs are bijections. They admit an inverse, but it
may be partial.

\begin{example}
	The category of sets and bijective functions is a dagger category. If the
	function $f \colon X \to Y$ is bijective, its dagger is $f\inv \colon Y \to
	X$.
\end{example}

\begin{definition}
	\label{def:dg-epi-mono}
	Let $\CC$ be a dagger category. A morphism $f \colon X \to Y$ is called:
	\begin{itemize}
		\item a \emph{dagger monomorphism} if $f\dg \circ f = \iid_X$;
		\item a \emph{dagger epimorphism} if $f \circ f\dg = \iid_Y$;
		\item a \emph{dagger isomorphism} if it is both a dagger
			monomorphism and a dagger epimorphism.
	\end{itemize}
\end{definition}

\begin{remark}
	A morphism that is dagger monic (resp.~epic) is in particular split monic
	(resp.~epic), and therefore is monic (resp.~epic). However, the converse is
	not true in general.
\end{remark}

To interpret computation, a category needs more structure. We assume that we
need to form  pairs of programs, represented by a monoidal structure $\otimes$,
and conditions on programs, represented by a monoidal structure $\oplus$,
\eg~the type of booleans given by $\one \oplus \one$.

A \emph{dagger rig category} is a dagger category equipped with symmetric
monoidal structures $(\otimes, I)$ and $(\oplus, O)$, such that $\otimes$ and
$\oplus$ are dagger functors, with their coherence isomorphisms, and with
additional natural dagger isomorphisms (satisfying coherence
conditions~\cite{laplaza1972coherence}):
\[
	\begin{array}{c}
		(X \oplus Y) \otimes Z \stackrel{\thicksim}{\longrightarrow} (X \otimes Z) \oplus (Y \otimes Z),
		\quad
		O \otimes X \stackrel{\thicksim}{\longrightarrow} O,
		\\[1.5ex]
		Z \otimes (X \oplus Y) \stackrel{\thicksim}{\longrightarrow} (Z \otimes X) \oplus (Z \otimes Y),
		\quad
		X \otimes O \stackrel{\thicksim}{\longrightarrow} O.
	\end{array}
\]

\begin{example}
	\label{ex:rel}
	The category of sets and relations is a dagger rig category. If $f \colon X
	\to Y$ is a relation between two sets, then $f\dg = \{ (y,x) \mid (x,y) \in
	f \}$. This category is known to be a model of linear logic \cite[Section
	2.1]{ehrhard2012scott} and of programming languages based on linear logic.
\end{example}

\begin{example}
	\label{ex:pinj}
	We write $\PInj$ for the category of sets and partial injections.
	It is an inverse category \cite{kastl1979inverse}, which means that it is
	equipped with a dagger and some equational conditions on the dagger.  An
	example of a partial injection and its partial inverse is given in the
	introduction (\ref{eq:pinj-fun}). Inverse categories have been generalised
	to \emph{restriction categories}
	\cite{cockett2002restriction-I,cockett2003restriction-II,cockett2007restriction-III}. This family of categories is used to represent
	reversible computation, \eg~as a model~\cite{kaarsgaard2021join} of
	Rfun~\cite{yokoyama2011reversible}, or as a model~\cite{nous2021invcat,nous2024invrec} of reversible pattern matching~\cite{sabry2018symmetric}.
\end{example}

\begin{example}
	\label{ex:hilb}
	Another example is the category of Hilbert spaces and linear maps, that we
	denote $\Hilb$. Its wide subcategory whose morphisms are unitaries (or
	dagger isomorphisms) is used to interpret pure quantum operations
	\cite{sabry2018symmetric,heunen2022information,carette2023quantum,carette2024easy,me-thesis}, and its wide subcategory whose morphisms are
	isometries (or dagger monomorphisms) is a model of quantum states~\cite[{\S
	3.5}]{me-thesis}. Both are also rig dagger categories.
\end{example}

The categories $\PInj$ and $\Hilb$ are $\oplus$-semicartesian. This means that
the unit of the monoidal structure associated to $\oplus$ is a terminal object
-- and therefore a zero object, since we are working within dagger categories.
This fits the story of the models for guarded recursion where a terminal object
is necessary to define the \emph{later} functor (see
Definition~\ref{def:later}) and the natural transformation \emph{next} (see
Definition~\ref{def:next}).

%%%%%%%%%%%%%%%%%%%%%%%%%%%%%
\subsection{The Guarded Construction from a Dagger Rig Category}
\label{sub:rig-construction}
%%%%%%%%%%%%%%%%%%%%%%%%%%%%%

Fix $\CC$ a $\oplus$-semicartesian dagger rig category, \ie~the monoidal unit
$O$ of $\oplus$ is a zero object. Thus there is a unique $!_X \colon X \to O$
and a unique $\inim_X \colon O \to X$ for all objects $X$; therefore, we have
$!_X \circ \inim_X = \iid_O$ and $(\inim_X)\dg =~ !_X$. The morphism $!_X
\colon X \to O$ is thus a dagger epimorphism. We obtain the left injection
$\iota_1 \colon X \to X \oplus Y$ by:
\(
	\begin{tikzcd}[column sep=large]
		X & X \oplus O & X \oplus Y,
		\arrow["(\rho^\oplus_X)^{-1}", from=1-1, to=1-2]
		\arrow["\iid_X \oplus \inim_Y", from=1-2, to=1-3]
	\end{tikzcd}
\)
and the left injection is therefore a dagger epimorphism as well.

\begin{lemma}
	\label{lem:rig-nq}
	The category $\QQ$ is a rig category. It is in particular $\oplus$-semicartesian.
\end{lemma}

We detail the interaction between the guarded structure and the rig structure
from $\CC$.

\begin{lemma}
	\label{lem:rig-guarded}
	If $X$ and $Y$ are objects in $\QQ$, then we have, with $\star \in \{ \otimes,
	\oplus \}$:
	\[
		\begin{array}{c}
			L^\CC (X {\star} Y) \cong L^\CC X \star L^\CC Y,
			\quad
			\nu^\CC_{X \star Y} \cong \nu^\CC_X \star \nu^\CC_X.
		\end{array}
	\]
\end{lemma}

\begin{corollary}
	\label{cor:rig-guarded}
	Let $X_1$ and $X_2$ be objects in $\QQ$. We have that
	\(
		\nu^\CC_{X_1 \oplus X_2} \circ \iota_i \cong \iota_i \circ \nu^\CC_{X_i}.
	\)
\end{corollary}

However, the category $\QQ$ does not handle the dagger well, as highlighted below.

\begin{remark}[Dagger in Guarded Construction]
	\label{rem:dagger}
	Morphisms in $\QQ$ are natural transformations whose components are
	morphisms in $\CC$. In that regard, the category $\QQ$ inherits some of the
	structure of $\CC$, but not all, \eg~the dagger. We however stick to the
	notation $f\dg$ for the \emph{componentwise} dagger of $f \colon X \to Y$
	in $\QQ$, even if it might not be a morphism in $\QQ$.
\end{remark}

In order to characterise which morphisms admit a dagger, we choose to work in a
less general category, defined as follows. First, we write $\CC_e$ for the wide
subcategory of $\CC$ whose morphisms are only dagger epimorphisms. We then
define the category $\NN$, whose objects are the objects of
$(\CC_e)^{\N^{\op}}$ and whose morphisms are natural transformations in $\CC$.
In other words, the category $\NN$ has cochains of dagger epimorphisms as
objects and arbitrary natural transformations as morphisms. This means that the
objects of $\NN$ are cochains of \emph{larger and larger} objects. Note that
$\NN$ is fully embedded in $\QQ$.

Similarly to $\QQ$, the category $\NN$ yields an $\Scat$-enriched category.
% with the following data:
% \begin{itemize}
% 	\item objects are objects in $(\CC_e)^{\N\op}$;
% 	\item hom-objects $\NN(X,Y)$ of $\Scat$, defined as: $\NN(X,Y)(n) = \{
% 			(f_0, \dots, f_n) \mid f \colon X \to Y \text{ in } \CC^{\N\op}
% 		\}$, a set of truncated natural transformations, and $\NN(X,Y)(n{+}1)
% 		\to \NN(X,Y)(n)$ is given as the function that forgets the last
% 		element;
% 	\item for all objects $X$, a morphism $\iid_X \colon 1 \to \NN(X,X),$
% 		which outputs truncated identity natural transformations;
% 	\item for all objects $X, Y, Z$, a morphism:
% 		\( \comp_{X,Y,Z} \colon \NN(Y,Z) \times \NN(X,Y) \to \NN(X,Z), \)
% 		which composes the truncated natural transformations.
% \end{itemize}
However, this presentation of the enrichment does not account for the dagger.
We choose to capture a more precise enrichment for $\NN$.

\begin{lemma}
	\label{lem:e-enriched}
	Due to the new restriction to dagger epimorphisms in the cochains, we
	can equivalently provide a description of the enrichment of $\NN$ in
	$\Scat$ with:
	\[
		\NN(X,Y)(n) = \{ f_n \mid f \colon X \to Y \text{ in } \CC^{\N\op} \},
		\quad
		\text{and}
		\quad
		\left\{ \begin{array}{ccl}
			\NN(X,Y)(n{+}1) & \to & \NN(X,Y)(n) \\
			f & \mapsto & r^Y_n \circ f \circ (r^X_n)\dg
		\end{array} \right.
	\]
\end{lemma}

\begin{remark}
	\label{rem:embedding-S}
	The embedding $E \colon \NN \embed \QQ$ is $\Scat$-enriched.
\end{remark}

We later use this category as a model for a guarded reversible programming
language (see \secref{sec:application}). To do so, the category $\NN$ needs to
preserve structure both linked to guarded recursion and the monoidal tensors.
The interaction with the dagger is detailed in the next section.

Note that the later functor $L^\CC$ (co)restricts to $\NN$ (if $X$ is a cochain
of dagger epimorphisms, $L^\CC X$ is too). We can then keep the notation $L^\CC
\colon \NN \to \NN$ non ambiguously.

\begin{lemma}
	\label{lem:e-later-enriched}
	The functor $L^\CC \colon \NN \to \NN$ is $\Scat$-enriched.
\end{lemma}

In the same way as $\QQ$ (see Lemma~\ref{lem:rig-nq}), the category $\NN$
inherits monoidal structures from the underlying category~$\CC$, and the
resulting monoidal structures in $\NN$ are pointwise.

\begin{lemma}
	\label{lem:rig-nq-e}
	The category $\NN$ is a rig category. It is in particular $\oplus$-cartesian.
\end{lemma}

We then show that $\NN$ has an interaction with the dagger in $\CC$, and thus
contains more structure inherited from $\CC$. We can then point out morphisms
in $\NN$ that can be reversed.

%%%%%%%%%%%%%%%%%%%%%%%%%%%%%
\subsection{Guarded Dagger Category}
\label{sub:dagger-guarded-dagger}
%%%%%%%%%%%%%%%%%%%%%%%%%%%%%

As explained in Remark~\ref{rem:dagger}, thus the dagger notation $(-)\dg$ in
the guarded construction is loose. However, this notation can be used in some
cases (see Lemma~\ref{lem:dagger-use}). We then point out which morphisms in
$\NN$ are such that their pointwise dagger is also a morphism in $\NN$.

\begin{lemma}
	\label{lem:dagger-use}
	If $f \colon X \to Y$ is in $\NN$, then $\nu^\CC_Y \circ f \circ \left(
	\nu^\CC_X \right)\dg \colon LX \to LY$ is in $\NN$.
\end{lemma}
% \begin{proof}[Proof of Lemma~\ref{lem:dagger-use}]
% %	Remember that $\nu$ is defined as $\nu_{X,n} = r^X_{n-1}$.
% 	Let us proceed by proving that $\nu^\QQ_Y \circ f \circ \left( \nu^\QQ_X
% 	\right)\dg$ is a natural transformation in $\CC$.
% 	\begin{align*}
% 		r_n^Y \circ f_{n{+}1} \circ (r^X_n)\dg \circ r^X_n
% 		&= f_n \circ r^X_n \circ (r^X_n)\dg \circ r^X_n \\
% 		&= f_n \circ r^X_n \\
% 		&= r^Y_n \circ f_{n{+}1} \\
% 		&= r^Y_n \circ f_{n{+}1} \circ r^X_{n{+}1} \circ (r^X_{n{+}1})\dg \\
% 		&= r^Y_n \circ r^Y_{n{+}1} \circ f_{n+2} \circ (r^X_{n{+}1})\dg.
% 	\end{align*}
% 	We have a natural transformation and $\nu^\QQ_Y \circ f \circ \left( \nu^\QQ_X
% 	\right)\dg\colon LX \to LY$ is a morphism in $\QQ$.
% \end{proof}

By Lemma~\ref{lem:later-enriched-view}, \emph{later} at the level of a hom-object
$\QQ(X,Y)$ is equivalent to \emph{later} applied to both objects $X$ and $Y$. A
similar statement holds in $\NN$ for \emph{next} due to the dagger.

\begin{lemma}
	\label{lem:next-enriched-view}
	Given $X$ and $Y$ two objects of $\NN$, we have
	$\nu^\Set_{\NN(X,Y)} = \nu^\CC_Y \circ - \circ \left( \nu^\CC_X \right)\dg$.
\end{lemma}
%\begin{proof}
%	This follows from Lemma~\ref{lem:enriched},
%	Lemma~\ref{lem:later-enriched-view} and Definition~\ref{def:next}.
%\end{proof}

This also means that if $f \colon X \to Y$ is a morphism in $\NN$, then
$\nu^\Set_{\NN(X,Y), n{+}1} (f_{n{+}1}) = f_n$.

The category $\NN$ is not a dagger category, but if it were, the dagger of a
morphism in $\NN$ would necessarily be the dagger taken pointwise on the
components in $\CC$ -- the same way the inverse of an isomorphism is pointwise.

\begin{definition}[Daggerable]
	If $f \colon X \to Y$ is a morphism in $\NN$, we say that it \emph{admits a
	dagger} or is \emph{daggerable} if the family $\{ f\dg_n \}_{n \in \N}$ of
	morphisms in $\CC$ verifies:
	\[
		\begin{tikzcd}
			Y(0) & Y(1) & Y(2) & Y(3) & \cdots \\
			X(0) & X(1) & X(2) & X(3) & \cdots
			\arrow["r^Y_0"', from=1-2, to=1-1]
			\arrow["r^Y_1"', from=1-3, to=1-2]
			\arrow["r^Y_2"', from=1-4, to=1-3]
			\arrow[from=1-5, to=1-4]
			\arrow["r^X_0", from=2-2, to=2-1]
			\arrow["r^X_1", from=2-3, to=2-2]
			\arrow["r^X_2", from=2-4, to=2-3]
			\arrow[from=2-5, to=2-4]
			\arrow["f\dg_0", from=1-1, to=2-1]
			\arrow["f\dg_1", from=1-2, to=2-2]
			\arrow["f\dg_2", from=1-3, to=2-3]
			\arrow["f\dg_3", from=1-4, to=2-4]
		\end{tikzcd}
	\]
	and therefore yields a morphism in $\NN$, that we write $f\dg \colon Y \to
	X$. We write $\QQdag$ for the wide subcategory of $\NN$ whose morphisms are
	daggerable. It is by definition a rig category.
\end{definition}

\begin{example}
	Morphisms $\nu^\CC_X$ are, in general, not daggerable.
\end{example}

While it is a difficult task to exactly characterise the morphisms of $\QQdag$,
we can point out a sufficient number of them to interpret a reversible
programming language.

\begin{lemma}
	\label{lem:daggerable-isom}
	Morphisms in $\NN$ with dagger isomorphisms components are daggerable.
\end{lemma}

This implies, in particular, that identity morphisms are in $\QQdag$ --
although this is already implied above when we refer to $\QQdag$ as a
subcategory of $\NN$. More interestingly, coherence morphisms associated to the
rig structure are defined pointwise, and therefore all their components are
dagger isomorphisms. Thus, Lemma~\ref{lem:daggerable-isom} is sufficient to
ensure that the category $\QQdag$ is a rig category, with the structure
inherited from the underlying category $\CC$.

\begin{lemma}
	\label{lem:daggerable-rig}
	If $f \colon X_1 \to Y_1$ and $g \colon X_2 \to Y_2$ are morphisms in $\QQdag$,
	then:
	\begin{itemize}
		\item $f \otimes g \colon X_1 \otimes X_2 \to Y_1 \otimes Y_2$ and
		\item $f \oplus g \colon X_1 \oplus X_2 \to Y_1 \oplus Y_2$ are
			morphisms in $\QQdag$.
	\end{itemize}
\end{lemma}

The lemmas above show that $\QQdag$ is a dagger rig category. It is therefore a
suitable model to interpret a simply-typed reversible programming language, and
the following lemmas continue to push towards the semantics of a reversible
language as expressive as possible.

\begin{lemma}
	\label{lem:zero-dag}
	The morphisms $!_X \colon X \to O$ and $\inim_X \colon O \to X$ in $\NN$
	admit a dagger.
\end{lemma}

This ensures that injections $\iota_i \colon X_i \to X_1 \oplus X_2$, with $i
\in \{1, 2\}$, are daggerable, because they are formed as composition of
daggerable morphisms (see \secref{sub:rig-construction}).

Following Lemma~\ref{lem:dagger-use}, we can point out another important family
of daggerable morphisms.

\begin{lemma}
	\label{lem:daggerable-use}
	If $f \colon X \to Y$ in $\NN$ is daggerable, then $\nu^\CC_Y \circ f \circ
	\left( \nu^\CC_X \right)\dg \colon LX \to LY$ is also.
\end{lemma}

The category $\QQdag$ is dagger rig and $\oplus$-semicartesian due to the
observations above. This category is also enriched in $\Scat$, which means that
the guarded fixed point operator can be used at the level of morphisms, with
morphisms that are reversible, since they admit a dagger.

We have outlined a relevant structure for guarded recursion with daggers, and
shown its link with the topos of trees $\Scat$ through an enrichment. We show
in the following section that this enrichment of $\NN$ in $\Scat$ also provides
fixed points for a certain class of functors.
These fixed points of functors are useful in the denotational semantics of
infinite data types.

%%%%%%%%%%%%%%%%%%%%%%%%%%%%%%%%%%%%%%%%%%%%%%%
\section{Application: Semantics of Guarded Symmetric Pattern matching}
\label{sec:application}
%%%%%%%%%%%%%%%%%%%%%%%%%%%%%%%%%%%%%%%%%%%%%%%

Symmetric pattern matching \cite{sabry2018symmetric} is a typed reversible
language, based on Theseus~\cite{james2014theseus}, first introduced as a
general framework for pure quantum programming -- \ie~quantum computing with
only reversible operations -- which works with infinite data types, such as
lists (which is unique in the pure quantum literature). The syntax of symmetric
pattern matching is simple and close to the type system, and we see it as the
$\lambda$-calculus for reversible programming.

There are several  refinements of the original paper \cite{sabry2018symmetric},
such as connections with linear logic and infinitary
proofs~\cite{chardonnet2023curry}, a full denotational semantics and adequacy
for its higher-order classical reversible fragment~\cite{nous2021invcat,nous2024invrec}, and a full denotational semantics and completeness result for
its first-order quantum fragment~\cite[Chapter 3]{me-thesis}. The denotational
semantics of the full higher-order quantum language is an open question (see
\cite[Section~5.2]{me-thesis}).

The required structure to be an interpretation of higher-order pattern matching
is:
\begin{enumerate}%[leftmargin=32pt, label=\textbf{(\arabic*)}]
	\item \label{rec:rig} a rig structure, interpretating tensors
		and sum types;
	\item \label{rec:para} parameterised fixed points for functors, allowing
		for guarded inductive types;
	\item \label{rec:join} a \emph{join} structure, to interpret iso
		abstractions, the first-order functions of the language;
	\item \label{rec:dag} a dagger structure, to account for the reversibility
		of first-order functions of the language;
	\item \label{rec:enri} an enrichment in a cartesian closed category,
		permitting a higher-order language with a (guarded) $\lambda$-calculus
		on top of functions;
	\item \label{rec:fix} a fixed point operator in the enrichment category, to
		interpret (guarded) recursion.
\end{enumerate}
We detail all those points below, with the corresponding syntax adapted to
guarded recursion. The full syntax and typing rules are detailed in
Figures~\ref{fig:syntax} and \ref{fig:typing}. We fix a $\oplus$-semicartesian
dagger rig category $\CC$; we write $\NN$ for the category of sequences of
dagger epimorphisms as objects and arbitrary natural transformations, and we
write $\QQdag$ for its full subcategory of morphisms that admit a dagger.

\begin{figure}
	\[
		\begin{array}{llcl}
			\text{(Ground types)} & A,B & ::= & \one \mid A \oplus B
			\mid A \otimes B \mid \later A \mid X \mid \mu X . A  \\
			\text{(Function types)} & T_1,T_2 & ::= & A \iso B \mid \later T \mid T_1 \to T_2 \\[1ex]
			\text{(Pairing)} & t,t_1,t_2 & ::= & * \mid t_1 \otimes t_2 \\
			\text{(Injections)} &&& \mid \inl t \mid \inr t \\
			\text{(Function application)} &&& \mid \omega~t \\
			\text{(Later modality)} &&& \mid \mathtt{next}~t \\
			\text{(Inductive terms)} &&& \mid \mathtt{fold}~t \\[1ex]
			\text{(Abstraction)} & \omega & ::=
			& \{ t_1 \mapsto t'_1 \mid \cdots \mid t_m \mapsto t'_m \} \\
			\text{(Higher modality)} &&& \mid \mathtt{next}~\omega \\
			\text{(Fixed points)} &&& \mid \phi \mid \fix \phi . \omega \\
			\text{(Higher functions)} &&& \mid \lambda \phi . \omega \mid \omega_2 \omega_1
		\end{array}
	\]
	\caption{Syntax of guarded symmetric pattern matching.}
	\label{fig:syntax}
\end{figure}

\begin{figure}
	\[
		\begin{array}{c}
			\infer{\Theta, X \vdash X}{}
			\quad
			\infer{\Theta \vdash \one}{}
			\quad
			\infer{\Theta \vdash \later A}{\Theta \vdash A}
			\quad
			\infer[\star {\in} \set{\oplus,\otimes}]{\Theta \vdash A \star B}{
				\Theta \vdash A
				&
				\Theta \vdash B
			}
			\quad
			\infer{\Theta \vdash \mu X . A}{\Theta,X \vdash A & X\text{ guarded in }A}
			\\[3ex]
			\infer{\entailiso A \iso B}{\vdash A & \vdash B}
			\qquad
			\infer{\entailiso \later T}{\entailiso T}
			\qquad
			\infer{\entailiso T_1 \to T_2}{\entailiso T_1 & \entailiso T_2}
			\\[3ex]
			\infer{
				\Psi ; \cdot \vdash * \colon I
			}{}
			\qquad
			\infer{
				\Psi ; \Delta_1, \Delta_2 \vdash t_1 \otimes t_2 \colon A_1 \otimes A_2
			}{
				\Psi ; \Delta_1 \vdash t_1 \colon A_1
				&
				\Psi ; \Delta_2 \vdash t_2 \colon A_2
			}
			\\[1.5ex]
			\infer{
				\Psi ; \Delta_2, \Delta_1 \vdash \letv{x \otimes y}{t_1}{t_2} \colon C
			}{
				\Psi ; \Delta_1 \vdash t_1 \colon A \otimes B
				&
				\Psi ; \Delta_2, x \colon A, y \colon B \vdash t_2 \colon C
			}
			\qquad
			\infer{
				\Psi ; \Delta \vdash \nnext t \colon {\later} A
			}{
				\Psi ; \Delta \vdash t \colon A
			}
			\\[1.5ex]
			\infer{
				\Psi ; \Delta \vdash \ini t \colon A_1 \oplus A_2
			}{
				\Psi ; \Delta \vdash t \colon A_i
			}
			\qquad
			\infer{
				\Psi ; \Delta \vdash \fold t \colon \mu X . A
			}{
				\Psi ; \Delta \vdash t \colon A[\mu X .A / X]
			}
			\\[1.5ex]
			\infer{
				\Psi ; \Delta \vdash \omega \circledcirc t \colon {\later} B
			}{
				\Psi ; \entailiso \omega \colon {\later} (A \iso B)
				&
				\Psi ; \Delta \vdash t \colon {\later} A
			}
			\qquad
			\infer{
				\Psi ; \Delta \vdash \omega~t \colon B
			}{
				\Psi ; \entailiso \omega \colon A \iso B
				&
				\Psi ; \Delta \vdash t \colon A
			}
			\\[2ex]
			\infer{
				\Psi ; \entailiso \{ t_1 \iso t'_1 \alt \cdots \alt t_n \iso t'_n \} \colon A \iso B
			}{
				\Psi ; \Delta_i \vdash t_i \colon A
				&
				\forall i \neq j, t_i~\bot~t_j
				&
				\Psi ; \Delta_i \vdash t'_i \colon B
				&
				\forall i \neq j, t'_i~\bot~t'_j
			}
			\\[1.5ex]
			\infer{
				\Psi ; \phi \colon T \entailiso \phi \colon T
			}{}
			\qquad
			\infer{
				\Psi \entailiso \nnext \omega \colon \later T
			}{
				\Psi \entailiso \omega \colon T
			}
			\qquad
			\infer{
				\Psi \entailiso \ffix \phi . \omega \colon T
			}{
				\Psi, \phi \colon \later T \entailiso \omega \colon T
			}
			\\[1.5ex]
			\infer{
				\Psi \entailiso \lambda \phi . \omega \colon T_1 \to T_2
			}{
				\Psi, \phi \colon T_1 \entailiso \omega \colon T_2
			}
			\qquad
			\infer{
				\Psi \entailiso \omega_2 \omega_1 \colon T_2
			}{
				\Psi \entailiso \omega_1 \colon T_1
				&
				\Psi \entailiso \omega_2 \colon T_1 \to T_2
			}
		\end{array}
		\]
		\caption{Typing rules of guarded symmetric pattern matching.}
		\label{fig:typing}
	\end{figure}

%%%%%%%%%%%%%%%%%%%%%%%%%%%%%
\subsection{Semantics of Ground Types}
\label{sub:sem-ground-types}
%%%%%%%%%%%%%%%%%%%%%%%%%%%%%

For the interpretation of types, we use the category $\NNe$, which is the
guarded construction arising from $\CC_e$, the subcategory of $\CC$ whose
morphisms are all dagger epimorphisms.  Note that $\NNe$ is a wide subcategory
of $\NN$, in the sense that they have the same objects.

We tackle Point~\eqref{rec:rig} in Lemma~\ref{lem:rig-nq} above. This allows
for the introduction of sum types $\oplus$ and tensor products $\otimes$.
Moreover, Theorem~\ref{th:param} ensures that locally contractive functors
admit parameterised fixed points, fitting Point~\eqref{rec:para}. Therefore,
inductive types $\mu X .  A$ can be formed, under the condition that $X$ is
guarded in $A$ (see \secref{sub:classic-syntax}). Together with the later
functor, we have a suitable model for the type introduction rules below.
\begin{equation*}
    \label{eq:type-formation}
    \begin{array}{c}
        \infer{\Theta, X \vdash X}{}
        \quad
        \infer{\Theta \vdash \one}{}
        \quad
        \infer{\Theta \vdash \later A}{\Theta \vdash A}
        \quad
		\infer[\star {\in} \set{\oplus,\otimes}]{\Theta \vdash A \star B}{
            \Theta \vdash A
            &
            \Theta \vdash B
        }
		\quad
		\infer{\Theta \vdash \mu X . A}{\Theta,X \vdash A & X\text{ guarded in }A}
    \end{array}
\end{equation*}
The interpretation of a type judgement $\Theta \vdash A$ is a functor
$\sem{\Theta \vdash A} \colon (\NNe)^{\abs\Theta} \to \NNe$. The denotation of
$\sem{\Theta \vdash \later A}$ is obtained by postcomposing $\sem{\Theta
\vdash A}$ with the later functor $L^{\CC_e} \colon \NNe \to \NNe$, and the one of
$\sem{\Theta \vdash \mu X . A}$ is $\sem{\Theta,X \vdash A}\nnoma$ (see
Theorem~\ref{th:param}).

\begin{remark}
	Components of morphisms in $\NNe$ are dagger epimorphisms; and isomorphic
	dagger epimorphism are dagger isomorphisms, the morphisms obtained by
	Theorem~\ref{th:param} then have dagger isomorphic components.
	Lemma~\ref{lem:daggerable-isom} shows that they are even morphisms in
	$\QQdag$.
\end{remark}

Terms admit only closed types.
The interpretation of a closed type $\cdot \vdash A$ is a functor $1 \to \NNe$,
therefore it is simply an object of $\NNe$ (and thus an object of $\NN$ and
$\QQdag$ too). We abuse notation and write $\sem A$ for the object of
$\NN$ given by $\sem{\cdot \vdash A}(*)$.

\begin{example}
	The interpretation of the unit type $\one$ is the object of $\NN$,
	represented in $\CC$ by:
	\(
		\begin{tikzcd}
			I & I & I & I & \cdots
			\arrow["\iid_I"', from=1-2, to=1-1]
			\arrow["\iid_I"', from=1-3, to=1-2]
			\arrow["\iid_I"', from=1-4, to=1-3]
			\arrow[from=1-5, to=1-4]
		\end{tikzcd}
	\)

% 	The type of guarded natural numbers is $\mu X . \one \oplus \later X$.
% 	Its interpretation in $\NN$ is given in $\CC$ by:
% 	\(
% 		\begin{array}{l}
% 			\begin{tikzcd}
% 				I & I \oplus I & (I \oplus I) \oplus I &
% 				((I \oplus I) \oplus I) \oplus I & \cdots
% 				\arrow["\iota_1\dg"', from=1-2, to=1-1]
% 				\arrow["\iota_1\dg"', from=1-3, to=1-2]
% 				\arrow["\iota_1\dg"', from=1-4, to=1-3]
% 				\arrow[from=1-5, to=1-4]
% 			\end{tikzcd}
% 		\end{array}
% 	\)

	Given a closed type $A$, the type of lists of type $A$ is $\mu X . \one
	\oplus (A \otimes \later X)$. Its interpretation in $\NN$ is:
	\(
	\begin{tikzcd}
		I & I \oplus \sem A(1) & (I \oplus \sem A(2)) \oplus \sem A(2)
		^{\otimes 2} & \dots
		\arrow["\iota_1\dg"', from=1-2, to=1-1]
		\arrow["r \circ \iota_1\dg"', from=1-3, to=1-2]
		\arrow["r \circ \iota_1\dg"', from=1-4, to=1-3]
	\end{tikzcd}
	\)
\end{example}

%%%%%%%%%%%%%%%%%%%%%%%%%%%%%
\subsection{Semantics of Ground Terms}
%%%%%%%%%%%%%%%%%%%%%%%%%%%%%

The primary terms of symmetric pattern matching are as expected. Typing
judgements for terms have the form $\Delta \vdash t \colon A$, where $\Delta =
x_1 \colon A_1, \dots, x_m \colon A_m$ is a context, and $A$ is a closed type.
The semantics of $\Delta$ is the tensor product of the semantics of all its
components $A_i$.
\[
	\begin{array}{c}
		\infer{
			\cdot \vdash * \colon I
		}{}
		\qquad
		\infer{
			\Delta_1, \Delta_2 \vdash t_1 \otimes t_2 \colon A_1 \otimes A_2
		}{
			\Delta_1 \vdash t_1 \colon A_1
			&
			\Delta_2 \vdash t_2 \colon A_2
		}
		\qquad
		\infer{
			\Delta_2, \Delta_1 \vdash \letv{x \otimes y}{t_1}{t_2} \colon C
		}{
			\Delta_1 \vdash t_1 \colon A \otimes B
			&
			\Delta_2, x \colon A, y \colon B \vdash t_2 \colon C
		}
		\\[1.5ex]
		\infer{
			\Delta \vdash \ini t \colon A_1 \oplus A_2
		}{
			\Delta \vdash t \colon A_i
		}
		\qquad
		\infer{
			\Delta \vdash \omega~t \colon B
		}{
			\entailiso \omega \colon A \iso B
			&
			\Delta \vdash t \colon A
		}
		\qquad
		\infer{
			\Delta \vdash \nnext t \colon {\later} A
		}{
			\Delta \vdash t \colon A
		}
		\\[1.5ex]
		\infer{
			\Delta \vdash \omega \circledcirc t \colon {\later} B
		}{
			\entailiso \omega \colon {\later} (A \iso B)
			&
			\Delta \vdash t \colon {\later} A
		}
		\qquad
		\infer{
			\Delta \vdash \fold t \colon \mu X . A
		}{
			\Delta \vdash t \colon A[\mu X .A / X]
		}
	\end{array}
\]

\begin{example}
	If we have $\Delta \vdash t \colon [A]$ and $\Delta' \vdash h \colon A$,
	the list with head $h$ and tail $t$ is obtained in the syntax by $\fold{}
	\inr{} (h \otimes \nnext t)$.
\end{example}

The interpretation of a type judgement $\Delta \vdash t \colon A$ is a morphism
$\sem{\Delta \vdash t \colon A} \colon \sem\Delta \to \sem A$ in $\NN$.
% This
% morphism can be represented, similar to ones in the previous sections, as a
% diagram in $\CC$:
% \[
% 	\begin{tikzcd}
% 		\sem\Delta(0) & \sem\Delta(1) & \sem\Delta(2) & \cdots \\
% 		\sem A(0) & \sem A(1) & \sem A(2) & \cdots
% 		\arrow["r^{\sem\Delta}_0"', from=1-2, to=1-1]
% 		\arrow["r^{\sem\Delta}_1"', from=1-3, to=1-2]
% 		\arrow[from=1-4, to=1-3]
% 		\arrow["r^{\sem A}_0", from=2-2, to=2-1]
% 		\arrow["r^{\sem A}_1", from=2-3, to=2-2]
% 		\arrow[from=2-4, to=2-3]
% 		\arrow["\sem{t}_0"', from=1-1, to=2-1]
% 		\arrow["\sem t_1", from=1-2, to=2-2]
% 		\arrow["\sem t_2", from=1-3, to=2-3]
% 	\end{tikzcd}
% \]
%
The semantics of injections $\ini\!$ (resp. $\nnext\!$, $\fold\!$) is obtained
by postcomposing with $\iota_i \colon \sem{A_i} \to \sem{A_1} \oplus \sem{A_2}$
(resp. $\nu^\CC_{\sem A} \colon \sem A \to L^\CC \sem A$, $\phi^{\sem{X \vdash
A}}_{\sem{\mu X . A}} \colon \sem{A[\mu X . A / X]} \to \sem{\mu X . A}$). The
interpretation for tensoring two terms is simply the tensor of the semantics of
the terms. The semantics of the $\mathtt{let}$, destructor of the tensor
product, is given by $\sem{t_2} \circ (\iid_{\sem{\Delta_2}} \otimes
\sem{t_1})$.

Note that among the morphisms mentioned above, only $\nu^\CC$ is not
daggerable.

%%%%%%%%%%%%%%%%%%%%%%%%%%%%%
\subsection{Semantics of First-Order Functions}
\label{sub:fo-fun}
%%%%%%%%%%%%%%%%%%%%%%%%%%%%%

The functions in symmetric pattern matching are called \emph{isos} -- short for
`isomorphisms', even though they are not necessarily isomorphic in some
extensions of the language -- and are formed as the \emph{join} of several
reversible partial functions.
\[
	\infer{
		\entailiso \{ t_1 \iso t'_1 \alt \cdots \alt t_n \iso t'_n \} \colon A \iso B
	}{
		\Delta_i \vdash t_i \colon A
		&
		\forall i \neq j, t_i~\bot~t_j
		&
		\Delta_i \vdash t'_i \colon B
		&
		\forall i \neq j, t'_i~\bot~t'_j
	}
\]
% With a family of terms $\Delta_i \vdash t_i
% \colon A$ and a family of terms $\Delta_i \vdash t'_i \colon B$, under some
% relevant necessary conditions, we can form the iso $\{ t_1 \iso t'_1 \alt
% \cdots \alt t_n \iso t'_n \} \colon A \iso B$, which we call an \emph{iso
% abstraction} -- akin to a $\lambda$-abstraction.
The most important premise in the typing rule above is pairwise
orthogonality~\cite[Section 2.2]{sabry2018symmetric}, but the detail is not
relevant in this paper. The type of iso abstractions takes its semantics as
hom-objects $\QQdag(\sem A, \sem B)$ -- which is an object in $\Scat$ -- since
we want those functions to be reversible and their semantics to have a dagger.
The interpretation of each subfunction forming an iso is given by:
\(
	\sem{t_i \iso t'_i} = \sem{\Delta \vdash t'_i \colon B} \circ \sem{\Delta
	\vdash t_i \colon A}\dg.
\)
As observed above in Remark~\ref{rem:dagger}, the latter is not necessarily
well-defined. To overcome this issue, we introduce a symmetric binary relation
-- that we write $\bowtie$, say `have same depth as' -- in order to ensure that
an iso is well-defined.
First we introduce \emph{next-free} contexts, with $\nnext{}\! \notin t$, as:
\[
	\begin{array}{lcl}
		C[-] & ::= & -
		\mid  \ini C[-] \mid C[-] \otimes t \mid t \otimes C[-] \mid \fold C[-] \mid
		\omega~C[-] \\
		&& \mid \letv{x \otimes y}{C[-]}{t} \mid \letv{x \otimes y}{t}{C[-]}
	\end{array}
\]

\begin{definition}
	Given two terms $t,t'$, we have $t \bowtie t'$ if it can be derived with
	the rules below. When $t \bowtie t'$ can be derived, we say that $t$ and
	$t'$ have the \emph{same depth}. The relation $\bowtie$ is defined as the
	smallest symmetric relation on terms such that $* \bowtie *$ and:
	\[
		\begin{array}{c}
			\infer{
				C[t] \bowtie C'[t']
			}{
				t \bowtie t'
			}
			\quad
			\infer{
				\nnext t \bowtie \nnext t'
			}{
				t \bowtie t'
			}
			\quad
			\infer{
				\nnext{} \ini t \bowtie \ini{} \nnext t'
			}{
				t \bowtie t'
			}
			%\\[1.5ex]
			\quad
			\infer{
				\nnext{t_1 \otimes t_2} \bowtie \nnext t'_1 \otimes \nnext t'_2
			}{
				t_1 \bowtie t'_1
				&
				t_2 \bowtie t'_2
			}
		\end{array}
	\]
	with $C[-]$ and $C'[-]$ potentially different next-free contexts.
\end{definition}

\begin{lemma}
	\label{lem:defined-iso}
	If $\Delta \vdash t \colon A$, $\Delta \vdash t' \colon B$ and $t \bowtie
	t'$, then $\sem{t_i \iso t'_i}$ is in $\QQdag(\sem A,\sem B)$.
\end{lemma}

We then assume that there is a \emph{join} structure in $\CC$, similar to the
one in join inverse categories~\cite{axelsen2016join}, or similar to the one
used for contractions between Hilbert spaces~\cite[Lemma~3.37]{me-thesis}, or
such as the summability for models of linear logic~\cite[Definition
4.5]{ehrhard2024summability}. That is to say, if two parallel morphisms $f,g
\colon X \to Y$ are so-called \emph{compatible} -- the definition does not
matter here --, the join $f \vee g \colon X \to Y$ is also a morphism in $\CC$.
This join should be distributive with composition, \ie~we have $h\circ (f \vee
g) = hf \vee hg$ and $(f \vee g) \circ h' = fh' \vee gh'$, and it should also
be compatible with the dagger, \ie~we have $(f \vee g)\dg = f\dg \vee g\dg$.
Given these conditions, this join structure generalises to $\QQdag$: say
that two parallel morphisms $f,g \colon X \to Y$ in $\QQdag$ are compatible if
their components are pointwise compatible.  Naturality is ensured since the
join distributes with composition.  Condition~\eqref{rec:join} is thus satisfied.
The semantics of an iso abstraction is therefore defined in $\QQdag$ (thus
fitting Point~\eqref{rec:dag}) as:
\[
	\sem{\{ t_1 \iso t'_1 \alt \cdots \alt t_n \iso t'_n \} \colon A \iso B}
	=
	\bigvee_i \sem{t_i \iso t'_i} =
	\bigvee_i \sem{\Delta \vdash t'_i \colon B} \circ \sem{\Delta \vdash t_i
	\colon A}\dg.
\]

\begin{example}[Flip the first Boolean]
	Consider the following iso of type $[\one \oplus \one] \iso [\one \oplus
	\one]$ that maps a list of Booleans to a list of Booleans where the first
	element is flipped:
	\[
		\left\{\begin{array}{lcl}
			[~] & \iso &[~] \\
			(\inl *) :: \nnext t & \iso & (\inr *) :: \nnext t \\
			(\inr *) :: \nnext t & \iso & (\inl *) :: \nnext t
		\end{array}\right\}
	\]
% 	The interpretation of the term $(\inl *) :: \nnext t$, which is short for
% 	the term in the original syntax $\fold{}\inr{}((\inl{*})\otimes\nnext{t})$,
% 	is the morphism $\phi \circ \iota_2 \circ (\iota_1 \otimes (\nu^\CC \circ
% 	\sem t))$. The semantics of the partial function $(\inl *) :: \nnext t \iso
% 	(\inr *) :: \nnext t$ is:
% 	\[
% 		\sem{(\inl *) :: \nnext t \iso (\inr *) :: \nnext t} =
% 		\qquad \qquad \qquad
% 	\]
% 	\[
% 		\sem{(\inr *) :: \nnext t} \circ \sem{(\inl *) :: \nnext t}\dg
% 	\]
% 	which involves one $\nu^\CC$ from the right-hand side and one
% 	$(\nu^\CC)\dg$ from the left-hand side, therefore it is a morphism in $\NN$
% 	and admits a dagger, thus is a morphism in $\QQdag$.
\end{example}

%%%%%%%%%%%%%%%%%%%%%%%%%%%%%
\subsection{Semantics of Higher-Order Types}
%%%%%%%%%%%%%%%%%%%%%%%%%%%%%

Due to the enrichment of our categories (namely, $\NNe$, $\NN$ and $\QQdag$) in
$\Scat$, we add a simply-typed guarded $\lambda$-calculus on top of the
functions, in our syntax.
Until now, the only function type was $A \iso B$
given two closed term types $A$ and $B$. We then allow for \emph{functions of
functions}, with the following type grammar:
\(
	T~~::=~~A \iso B \alt \later T \alt T_1 \to T_2.
\)
These functions types take their semantics as objects in $\Scat$. The semantics
of $A \iso B$ is $\QQdag ( \sem A, \sem B )$. The semantics of $T_1 \to T_2$ is
$[\sem{T_1} \to \sem{T_2}]$, the exponential object in $\Scat$.

%%%%%%%%%%%%%%%%%%%%%%%%%%%%%
\subsection{Semantics of Higher-Order Terms and Functions}
%%%%%%%%%%%%%%%%%%%%%%%%%%%%%

We add \emph{function contexts} $\Psi = \phi_1 \colon T_1, \dots, \phi_k \colon
T_k$; now well-typed terms and functions can depend on a function context, with
term judgements such as $\Psi; \Delta \vdash t \colon A$ and function
judgements as $\Psi \entailiso \omega \colon T$.  An iso abstraction can be
applied to a function, thus $\Psi; \Delta \vdash \omega~t \colon B$ is
well-formed whenever $\Psi \entailiso \omega \colon A \iso B$ and $\Psi; \Delta
\vdash t \colon A$ are. The denotational semantics of term judgement is given
as a morphism in $\Scat$:
\(
	\sem{\Psi; \Delta \vdash t \colon A} \colon
	\sem\Psi \to \NN(\sem\Delta, \sem A).
\)
In addition, the grammar of functions is extended to a simply-typed $\lambda$-calculus.
\[
	\begin{array}{lcl}
		\omega & ::= & \{ t_1 \iso t'_1 \alt \cdots \alt t_n \iso t'_n \}
		\alt \phi
		\alt \lambda \phi . \omega \alt \omega_2 \omega_1 \alt \ffix \phi . \omega
		\alt \nnext \omega \alt \omega_2 \circledcirc \omega_1
	\end{array}
\]
The semantics for a function judgement $\Psi \entailiso \omega \colon T$ is a
morphism in $\Scat$ with type $\sem\Psi \to \sem T$. Therefore the semantics of
the new terms are the usual ones for a simply-typed $\lambda$-calculus in a
cartesian closed category, as pointed out in \eqref{rec:enri}.

The structure of $\Scat$ allows for the addition of several other kinds of
functions terms, inherited from the guarded lambda calculus: a delay operation
$\nnext \omega$, whose semantics is $\nu^\Set$ and a composition of delayed
operations $\omega_2 \circledcirc \omega_1$. A delayed iso $\omega \colon
\later (A \iso B)$ can also be applied to a delayed term $t \colon \later
A$ to obtain $\omega \circledcirc t \colon \later B$.

As shown in \secref{sub:topos}, there is a fixed point operator linked to the
guarded structure in $\Scat$. We then add to the syntax of functions a fixed
point operator $\ffix\!$, as desired in Point~\eqref{rec:fix}:
\begin{equation}
	\label{eq:type-fix}
	\infer{
		\Psi \entailiso \ffix \phi . \omega \colon T
	}{
		\Psi, \phi \colon \later T \entailiso \omega \colon T
	}
\end{equation}
The morphism $\sem{\Psi, \phi \colon \later T \entailiso \omega \colon T}$ is
contractive on its last variable, and therefore admits a guarded fixed
point~\cite[Theorem~2.4]{birkedal2012first}, used as the semantics for the term
$\Psi \entailiso \ffix \phi . \omega \colon T$.

\begin{example}
	Given $A$ and $B$ two types, we have now a proper semantics for the function
	$\mathtt{map}$ which applies a function $A \iso B$ to all the element of a
	list $[A]$.
	\[
		\mathtt{map} = \ffix \phi^{\later ((A \iso B) \to ([A] \iso [B]))} .
		\lambda \psi^{A \iso B} .
		\left\{ \begin{array}{lcl}
			[~] & \iso & [~] \\
			h :: t & \iso & (\psi~h) :: ((\phi \circledcirc \nnext \psi) \circledcirc t)
		\end{array}\right\}
	\]
\end{example}

%%%%%%%%%%%%%%%%%%%%%%%%%%%%%
\subsection{The Pure Quantum Case}
%%%%%%%%%%%%%%%%%%%%%%%%%%%%%

As mentioned above, the semantics of symmetric
pattern matching~\cite{sabry2018symmetric} with inductive types and
higher-order quantum operations is an open question. There are multiple factors
that make this question hard, \eg~categories of Hilbert spaces are unlikely to
interpret recursion (see \cite[Section~5.2]{me-thesis}). Guarded recursion
helps bypass those issues.

\begin{example}[Quantum control]
	In pure quantum computing, the type of quantum bits (or \emph{qubits}) is
	given by $\qubit \defeq \one \oplus \one$. The term $\inl * \colon \qubit$
	represents the qubit $\ket 0$, and $\inr *$, the qubit $\ket 1$. The
	high-order formalism of symmetric pattern matching allows for a general
	control operation $\mathtt{qctrl} \colon (A \iso B) \to (A \iso B) \to
	(\qubit \otimes A \iso \qubit \otimes B)$, that we refer to as
	\emph{quantum control} or \emph{quantum if}:
	\[
		\mathtt{qctrl} = \lambda \phi^{A \iso B} . \lambda \psi^{A \iso B} .
		\left\{ \begin{array}{lcl}
			\ket 0 \otimes y & \iso & \ket 0 \otimes (\phi~y) \\
			\ket 1 \otimes y & \iso & \ket 1 \otimes (\psi~y)
		\end{array} \right\}
	\]
	It is a \emph{quantum if}, because it applies either $\phi$ or $\psi$ to
	the second qubit depending on the value of the first qubit, without
	measuring it.
\end{example}

\begin{example}[General QFT]
	The \emph{Quantum Fourier Transform} is a subroutine that plays a central
	role in quantum algorithms. In the current literature, a quantum algorithm
	is often stated in the form of quantum circuits on a fixed number of
	qubits. We can here abstract away from this low-level description, and
	offer more appealing view for programming languages, by defining a quantum
	Fourier transform subroutine for lists of qubits.

	First, given an iso $\omega \colon A \iso A$, we define $\omega \circ
	\omega$ as the iso abstraction $\{ x \iso \omega~(\omega~x) \} \colon A
	\iso A$. We write $\mathtt 2$ for the type of qubits $\one \oplus \one$ for
	space reasons. We can then write the iso that applies gradual rotations to
	a list of qubits; here, it eventually applies $\psi^k$ to the
	$k^{\text{th}}$ element of the list (with $T = \later ((\two \iso \two) \to
	(\two \iso \two) \to ([\two] \iso [\two]))$):
	\[
		\mathtt{Rot} = \ffix \phi^T .
		\lambda \psi^{\two \iso \two} . \lambda \psi'^{\two \iso \two} .
		\left\{ \begin{array}{lcl}
			[~] & \iso & [~] \\
			h :: t & \iso & (\psi'~h) :: (\phi \circledcirc \nnext \psi
			\circledcirc \nnext (\psi \circ \psi')) \circledcirc t)
		\end{array}\right\}
	\]
	then, given $\mathtt{had} \colon \two \iso \two$ the Hadamard gate,
	we have (with $T' = \later ((\two \iso \two) \to ([\two] \iso [\two]))$):
	\[
		\scalebox{0.9}{
			$
		\mathtt{QFT} = \ffix \phi^{\later ((\two \iso \two) \to ([\two] \iso [\two]))} .
		\lambda \psi^{\two \iso \two} .
		\left\{ \begin{array}{lcl}
			[~] & \iso & [~] \\
			h :: \nnext t & \iso & \letv{h' \otimes t'}{\mathtt{qctrl}(\mathtt{Rot}(\psi))~((\mathtt{had}~h) \otimes t)}{\!} \\
			&& h' :: ((\phi \circledcirc \nnext \psi) \circledcirc \nnext t')
		\end{array}\right\}
			$
		}
	\]
	This offers a higher-order perspective on reversible quantum programming.
\end{example}

%%%%%%%%%%%%%%%%%%%%%%%%%%%%%
\subsection{What about streams?}
%%%%%%%%%%%%%%%%%%%%%%%%%%%%%

The story of reversible guarded recursion is very different to the classical
one. For example, streams are usually computed as a fixed point of the functor
$X \otimes T$, obtained as the limit of the following diagram, assuming $T$ is
a terminal object.
\[
	\begin{tikzcd}
		T & X \otimes T & X \otimes X \otimes T & \cdots
		\arrow["!"', from=1-2, to=1-1]
		\arrow["X \otimes !"', from=1-3, to=1-2]
		\arrow[from=1-4, to=1-3]
	\end{tikzcd}
\]
\begin{lemma}
	Given a dagger category $\CC$, if $T$ is a terminal object in $\CC$, then
	it is also initial.
\end{lemma}
Therefore, in a dagger rig category, a terminal object is also a zero object,
and the diagram above falls down to a cosequence of zero objects. whose limit
is also the zero object.
This does not necessarily imply that streams cannot be manipulated reversibly;
but the semantics of streams certainly cannot be computed in a usual way
in a dagger category. Working with reversible streams then requires a new
theory.

% %%%%%%%%%%%%%%%%%%%%%%%%%%%%%%%%%%%%%%%%%%%%%%%
% \section{Conclusion}
% \label{sec:guar-conclusion}
% %%%%%%%%%%%%%%%%%%%%%%%%%%%%%%%%%%%%%%%%%%%%%%%
%
% We showed that an arbitrary category, with a terminal object, can be elevated
% to a model of guarded recursion with a simple construction. This construction
% preserves monoidal structures. A relevant example of semantics that involves
% monoidal structures that are not cartesian is reversible programming, of which
% we give an interpretation in dagger rig categories. The dagger structure is not
% preserved by the guarded construction, but we show that the latter has a
% sufficiently large subcategory for recursive calls of reversible functions.
%
% There are several points that can be tackled as further work. One avenue of
% research is to determine the internal logic behind the categories $\QQ$, $\NN$
% and $\QQdag$, since their structure is close to the one of the topos of trees.
% The category $\QQdag$ could also be studied through the prism of axioms of
% categories~\cite{heunen2022axioms,heunen2024axioms,dimeglio2024dagger}.
%
% From a quantum programming language point of view, one could study quantum
% information effects~\cite{pablo2022universal,heunen2022information} with
% regard to the category $\QQdag$, which is a dagger rig category. We could
% therefore have a sound way of adding effects, such as measurement, to symmetric
% pattern matching.

%%%%%%%%%%%%%%%%%%%%%%%%%%%%%
\section*{Acknowledgments}
%%%%%%%%%%%%%%%%%%%%%%%%%%%%%

I would like to thank Benoît Valiron and Vladimir Zamdzhiev for their support
and helpful comments during my thesis; and thanks to Robert Booth, Titouan
Carette, Kostia Chardonnet, Pierre Clairambault, Jonas Frei, Laurent Regnier,
and Morgan Rogers for our chats about this topic. I am also grateful to Chris
Heunen, Robin Kaarsgaard and Kim Worrall, and more generally the Quantum
Programming Group in the University of Edinburgh for their support and
feedback.

\bibliographystyle{./entics}
\bibliography{ref}

\end{document}